\documentclass[submission,copyright,creativecommons]{eptcs}
\providecommand{\event}{MCU 2013} 
\usepackage{breakurl} 

\usepackage[pdftex]{graphicx}
\usepackage{latexsym, amsgen, amsmath, amstext, amsbsy, amsopn, amsfonts, amsthm, amssymb}
\newtheorem{lemma}{Lemma}
\newtheorem{theorem}{Theorem}
\newtheorem{corollary}{Corollary}

\newcommand{\tnMarkZero}{\ensuremath{0\hspace{-1.2ex}/\,}}
\newcommand{\tnMarkOne}{\ensuremath{1\hspace{-1.2ex}/\,}}

\newcommand{\tnMarkZeroB}{\ensuremath{\textbf{0}\hspace{-1.2ex}/\,}}
\newcommand{\tnMarkOneB}{\ensuremath{\textbf{1}\hspace{-1.2ex}/\,}}

\providecommand{\urlalt}[2]{\href{#1}{#2}}
\providecommand{\doi}[1]{doi:\urlalt{http://dx.doi.org/#1}{#1}}

\title{A Small Universal Petri Net}
\author{Dmitry A. Zaitsev
\institute{International Humanitarian University\\
Department of Computer Science and Innovation Technology}
\institute{Fontanskaya Doroga, 33, Odessa 65009, Ukraine}
\email{daze@acm.org}
}

\def\titlerunning{Small Universal Petri Net}
\def\authorrunning{D.A. Zaitsev}
\begin{document}
\maketitle

\begin{abstract}
A universal deterministic inhibitor Petri net with 14 places, 29 transitions and 138 arcs was constructed via simulation of Neary and Woods' weakly universal Turing machine with 2 states and 4 symbols; the total time complexity is exponential in the running time of their weak machine. To simulate the blank words of the weakly universal Turing machine, a couple of dedicated transitions insert their codes when reaching edges of the working zone. To complete a chain of a given Petri net encoding to be executed by the universal Petri net, a translation of a bi-tag system into a Turing machine was constructed. The constructed Petri net is universal in the standard sense; a weaker form of universality for Petri nets was not introduced in this work.
\end{abstract}

\section{Introduction}

Standard universal Turing machines (TM) simulate an arbitrary given TM using tape with all cells, except a finite working zone, filled with infinite repetition of a blank symbol ($\lambda$). In weakly universal Turing machines, an infinite repetition of definite blank words is written: $w_l$ to the left and $w_r$ to the right of the working zone. Such a modification leads to obtaining smaller universal machines. 

Neary and  Woods \cite{NE08,NW09W} improved on the work of Cook~\cite{C04}, to give the smallest weakly universal Turing machines (WUTM) with state-symbol pairs of (2,4), (3,3), and (6,2); moreover, these machines work in polynomial time. 

A universal Petri net (UPN) was constructed in \cite{ZV12UPN} using direct encoding of an arbitrary Petri net by a definite number of nonnegative integer variables which are loaded as a marking of definite places into the universal net. To obtain the small universal Petri net with 14 places and 42 transitions \cite{ZV13}, small universal TM of Neary and Woods (6,4) \cite{NW09FI} was simulated by deterministic inhibitor Petri net (DIPN). Besides, for Petri net encoding, a translation of DIPN into bi-tag system was constructed.

In the present paper, universal DIPN with 14 places and 29 transitions is presented as a result of WUTM(2,4) \cite{NE08,NW09W} simulation by DIPN. The simulation of weakness in TMs required an amendment of the approach used in \cite{ZV13} in order to produce the codes of the blank words during the simulation. Moreover, for a given DIPN encoding to be executed by simulated WUTM(2,4), a translation from bi-tag system to TM was constructed.

Constructed in the present paper universal Petri net is universal in the standard sense; a weaker from of universality for Petri nets was not introduced in this work.

\section{Basic Notations and Definitions}

\subsection{A deterministic inhibitor Petri net}

\textit{A deterministic inhibitor Petri net} (DIPN) \cite{ZV13} is a bipartite directed graph supplied with a dynamic process that has deterministic behavior. DIPN is denoted as a quadruple $N=(P,T,F,{\mu}^0)$ , where $P$ and $T$ are disjoint sets of vertices called places and transitions respectively, the mapping $F$ defines arcs between vertices, and the mapping ${\mu}^0$ represents the initial state (marking). The transition choice order is defined by their enumeration $T=\{t_1,t_2,...,t_n\}, n=|T|$; places are also supposed been enumerated $P=\{p_1,p_2,...,p_m\}, m=|P|$, and the places' marking is represented as a vector $\overline{\mu}, {\mu}_j={\mu}(p_j)$ with integer index $j, 0<j \le m$, where ${\mu}_j$ is a nonnegative integer equal to the number of tokens situated in place $p_j$. 

The mapping ${F:(P \times T)\rightarrow \{0,w^{-},-1\}\cup (T\times P)\rightarrow\{0,w^{+}\}}$ defines arcs, where $w^{-}$ and $w^{+}$ are positive integers, a zero value corresponds to the arc absence, a positive value -- to the regular arc with indicated multiplicity, and a minus unit -- to the inhibitor arc. As it was mentioned in \cite{ZV13}, Petri nets with multiple arcs are easily converted to \textit{ordinary} Petri nets with regular arcs' multiplicity equal to unit. 

In graphical form, places are drawn as circles and transitions as rectangles. An inhibitor arc is represented by a small hollow circle at its end, and a small solid circle (read arc) represents the abbreviation of a loop. Regular arc's multiplicity greater than unit is inscribed on it and place's marking greater than zero is written inside it.

The \textit{behavior} (dynamics) of a DIPN is described by the state equation of inhibitor Petri net \cite{ZV12UPN} supplemented by the condition of the firable transition choice having the minimal number. The present work considers the behavior as a result of sequential applying the following \textit{transition firing rule}: 

\begin{enumerate}

\item net transitions ${t_{i}}$ are checked sequentially with integer index $i$ ranging from $1$ to $n$;

\item transition ${t_{{i}}}$ is firable iff each place $p_j$ with a regular arc directed to transition ${t_{{i}}}$ contains at least $w^{-}_{j,i}$ tokens, and each place with an inhibitor arc directed to transition ${t_{{i}}}$ does not contain tokens (${\forall j:}$ ${\mu _{{j}}\ge w^{-}_{j,i}}$ when ${F(p_{{j}},t_{{i}})=w^{-}_{j,i}}$ and ${\mu_{{j}}=0}$ when ${F(p_{{j}},t_{{i}})=-1}$);

\item the first firable transition fires (with the minimal value of index ${i}$); 

\item when transition ${t_{{i}}}$ fires, it

\begin{enumerate}
\item extracts $w^{-}_{j,i}$ tokens from each its input place (for regular arcs) ${p_{{j}}:F(p_{{j}},t_{{i}}) \ge w^{-}_{j,i}}$;

\item puts $w^{+}_{k,i}$ tokens into each its output place ${p_{{k}}:F(t_{{i}},p_{{k}})=w^{+}_{k,i}}, w^{+}_{k,i}>0$;
\end{enumerate}

\item the net halts if firable transitions are absent.

\end{enumerate}

Step-by-step firing transitions consumes and produces tokens within their incidental places that looks rather like moving tokens between places of a Petri net. In a classical Petri net, an arbitrary firable transition is chosen to fire at a step that induces a nondeterministic character of their behavior. Even general nets with priorities are nondeterministic since equal priorities of firable transitions may occur. The behavior of a DIPN is restricted by the firing the firable transition with the minimal number. Thus, transitions numbers are interpreted as their priorities which are different inducing the deterministic character of a DIPN behavior. 

Petri nets are considered as a graphical canvas for concurrent programs, where a paradigm of computation is not restricted by an instructions sequence or a set of sequences, promising hyper-parallel implementations.

\subsection{Turing machines}

\textit{A Turing machine} (TM) is denoted as ${M=(\Omega,\Sigma,f,us,uh)}$, where ${\Omega}$ is the set of internal states, ${\Sigma}$ is the tape alphabet, $f$ is the transition function, $us$ is the start state, $uh$ is the halt state. 

A Turing machine consists of the tape infinite in both directions divided in cells and the control head that moves along the tape, reads and replaces symbols in cells. An instruction of a Turing machine is a quintuple ${(x,u,x',v,u')}$, where ${x}$ is the current cell symbol, ${u}$ is the current internal state, ${x'}$ is a new symbol that replaces ${x}$ in the current cell, ${v}$ denote one of the head moves either to the left ($left$) or to the right ($right$), or stand still ($stay$), ${u'}$ is a new internal state. Usually, instructions are represented by a table (matrix) with keys given by the pair ${(x,u)}$ and values written as triples of the form ${x'vu'}$. Initially, blank symbols are written into cells. Computation of a Turing machine consists in transforming an input word written on the tape into an output word read when the machine halts. The minimal part of the tape containing nonblank symbols is named a \textit{working zone}.

A universal Turing machine is a prototype of computers. It accepts as input any given Turing machine (program) and its tape word (data) and executes the program over the data. The crucial point is an encoding of a given program and data in an alphabet with limited number of symbols and employing a limited number of internal states. 

The following notation is adopted for minimal universal Turing machines \cite{NE08}: ${{UTM}(m,n)}$, where ${m}$ is the number of internal states and ${n}$ is the number of tape symbols (including the blank symbol). The size of the transition table of ${(m,n)}$-machine is ${m \times n}$; in case not all feasible instructions are employed, the actual number of instructions ${l\le m \times n}$ is an extra characteristic.

\subsection{A bi-tag system}

\textit{A bi-tag system} (BTS) \cite{NW09FI} is a quadruple ${B=(A,E,e_{{h}},R)}$, where ${A}$ and ${E}$ are disjoint finite sets of symbols (alphabets), ${e_{{h}} \in E}$ is the halt symbol and ${R}$ is the finite set of productions in one of three valid forms:

{\centering ${R(a)=a},~{R(e,a)\in {AE}},~{R(e,a) \in {AAE}}$, \par}

where ${a\in A}$, ${e\in E}$ and ${R}$ are defined on all elements of the set ${\left\{A\cup ((E-\{e_{{h}}\})\times A)\right\}}$ and undefined on all elements of the set ${\{e_{{h}}\}\times A}$; the bi-tag system is deterministic. A BTS \textit{configuration} is a word of the form ${w=A^{{*}}({EA}\cup {AE})A^{{*}}}$.

A BTS \textit{computation step} consists of the application of productions in one of two valid ways:
\begin{itemize}
\item if ${w=as'}$, then ${as'{\succ}s'R(a)}$,
\item if ${w=eas'}$, then ${eas'{\succ}s'R(e,a)}$.
\end{itemize}

A BTS \textit{computation} is a finite sequence of computation steps that are consecutively applied to an initial configuration. If ${e_{{h}}}$ is the leftmost symbol in the current configuration, the computation halts.

\section{Design of UPN(14,29)}

At first we choose the smallest known (weakly) universal TM for simulation. The two machines of Turlough Neary and Damien Woods \cite{NE08,NW09W} with state-symbol pairs of (2,4) and (3,3) use 8 instructions, so they are equal in size for our purpose -- they produce the same number of Petri net transitions when simulating TM transition function. We choose WUTM(2,4), since it is defined on all state-symbol pairs.

Apart from Turing machines, many other computationally universal systems are known such as Markov normal algorithms \cite{ZV12NAM}, tag systems \cite{LDM06}, cellular automata \cite{C04} etc. Progress in constructing smaller system for one of these models often evokes immediate progress for others. Manifold translations techniques among them also influence this process and are themselves an object of investigation aimed to reduction in size and complexity. Direct simulation and chains of translations approaches compete. Recent results \cite{NE08} show that direct simulation is rather efficient. Traditionally, translations into tag systems are applied to construct small universal TM. 

In the present paper WUTM(2,4) \cite{NE08,NW09W} is simulated by a DIPN via technique presented in \cite{ZV13}. But WUTM(2,4) itself simulates cellular automaton Rule~110 (CA) which was shown universal by Matthew Cook \cite{C04}. Its input should be encoded in cyclic tag system (CTS) \cite{C04}. To fill the gap in the chain of translating DIPN into cyclic tag system, a translation from bi-tag system to either cyclic tag system or 2-tag system is requited. But considering that cyclic tag system could be constructed directly on TM \cite{NE08,NW09W}, the following chain of translations was chosen

{\centering ${{DIPN}\rightarrow {BTS}\rightarrow {TM}}$. \par}

In the above chain, the translation ${{DIPN}\rightarrow {BTS}}$ was constructed in \cite{ZV13}, and the translation \break ${{BTS}\rightarrow {TM}}$ is constructed in the present paper. It is supposed, that afterwards, the translations\break  ${{TM}\rightarrow {CTS}}$ \cite{NE08,NW09W} and ${{CTS}\rightarrow {CA}}$ \cite{C04} work, transforming obtained TM into CA code of glides for execution by WUTM(2,4) which simulates cellular automaton behavior. So the complete chain looks like

{\centering ${{DIPN}\rightarrow {BTS}\rightarrow {TM}\rightarrow {CTS}\rightarrow {CA}}$. \par}

The technique in \cite{ZV13} simulates TM tape as a triple ${(L,X,R)}$ consisting of the current cell symbol code ${X}$ and codes of two stacks representing the left ${L}$ and right ${R}$ parts of the tape with respect to the current cell symbol. The blank symbol is encoded by zero so hitting the stack bottom is not recognized but produces a new blank symbol. Simulating weak TM by DIPN is rather different task.

In the present paper, the tape alphabet is encoded starting from unit via the function $s(x)$ given by column 2 of Table~\ref{tab:wutm24} while states of TM are encoded starting from zero via the function $s(u)$ given by row 2 of Table~\ref{tab:wutm24}; the radices for symbols and states encoding are denoted as ${{rX}=5}$ and ${{rU}=2}$ respectively. Symbol $"r"$ is used to denote a radix, and symbol $"s"$ -- to denote a code; notation like ${s(\alpha)}$ represents the code of an object ${\alpha }$, where ${\alpha}$ could be either a separate symbol or a chain of symbols written on the tape. Codes of separate symbols are given by Table~\ref{tab:wutm24} and chains of symbols are encoded as numbers in a positional (radix) notation

\begin{equation}
  s(x_{{l-1}}x_{{l-2}}{...}x_{{0}})=\sum_{i=0}^{l-1}{s(x_{{i}})\cdot r^{{i}}}.
\end{equation}

Table~\ref{tab:wutm24} represents WUTM(2,4) behavior \cite{NE08,NW09W} with respect to the chosen encoding; the notation $"{\Sigma {\backslash} \Omega}"$ means the title $"{\Sigma}"$ of the column and the title $"{\Omega}"$ of the row. Note that, in WUTM(2,4), ${\Omega=\{u_{{1}},u_{{2}}\}}$, ${\Sigma=\{0,1,{\tnMarkZero},{\tnMarkOne}\}}$, ${{us}=u_{{1}}}$ and the machine does not halt in usual sense, since it simulates Rule~110 computations with special halt conditions \cite{C04}.

\begin{table}
  \begin{center}
    \begin{tabular}{|c|c|c|c|}
      \hline
      $\Sigma \backslash \Omega$ &  & $u_1$ & $u_2$\\
      \hline
        & $s(\Sigma) \backslash s(\Omega)$ & 0 & 1\\
      \hline
      0 & 1 & $3,left,0$ & $4,right,0$\\
      \hline
      1 & 2 & $4,left,1$ & $3,left,1$\\
      \hline
      $\tnMarkZero$ & 3 & $4,left,0$ & $1,right,1$\\
      \hline
      $\tnMarkOne$ & 4 & $4,left,0$ & $2,right,1$\\
      \hline
    \end{tabular}
  \end{center}
  \caption{WUTM(2,4) behavior \cite{NE08,NW09W} and its encoding}
  \label{tab:wutm24}
\end{table}

Remind that in weakly universal Turing machines, an infinite repetition of definite blank words is written: $w_l$ to the left and $w_r$ to the right of the working zone. The left and right blank words \cite{NE08,NW09W} are ${w_{{l}}=  00{\tnMarkZero}1}$ and ${w_{{r}}=0{\tnMarkOne}{\tnMarkZero}{\tnMarkZero}0{\tnMarkOne}}$ respectively. According to (1), the codes of the left and right blank words are calculated as follows:

{\centering
${{sw}_{{l}}=s(w_{{l}})=((s(0)\cdot {rX}+s(0))\cdot {rX}+s({\tnMarkZero}))\cdot {rX}+s(1)=((1\cdot 5+1)\cdot 5+3)\cdot  5+2={167}},$

${{sw}_{{r}}=s(w_{{r}})=((((s({\tnMarkOne})\cdot {rX}+s(0))\cdot {rX}+s({\tnMarkZero}))\cdot {rX}+s({\tnMarkZero}))\cdot {rX}+s({\tnMarkOne}))\cdot {rX}+s(0)=}$

${=((((4\cdot 5+1)\cdot 5+3)\cdot 5+3)\cdot 5+4)\cdot 5+1={13596}}$. \par}

Following the stack order of encoding/decoding regarding the current cell symbol, the left blank word was encoded from left to right while the right blank word -- from right to left.

Encoding tape symbols starting from unit allows recognition the left and right edges of the tape working zone via encountering zero value of the corresponding code. Obtaining zero is possible for only one tape side at a time: for the left side after moving to the left and for the right side after moving to the right. Substituting obtained zero by the corresponding code ${{sw}_{{l}}}$ or ${{sw}_{{r}}}$ models infinite blank words.

\section{Simulating WUTM(2,4) by DIPN}

WUTM(2,4) is simulated by a DIPN using the technique presented in \cite{ZV13}; as a result UPN(14,29) was obtained which general scheme of work is shown in Fig.~\ref{fig:upn1429}. Subnets are depicted as rectangles with double line border. Some vertices have mnemonic names besides their numbers; place $U$ contains encoded TM state $s(u)$, place $X$ contains encoded current cell symbol $s(x)$ (according to Table~\ref{tab:wutm24}), and places $L$ and $R$ contain encoded via equation (1) left and right parts of the tape working zone respectively regarding the current cell. 

Comparing \cite{ZV13}, the following modifications and amendments have been done: 
\begin{itemize}
\item subnet $FS$ was constructed on Table~\ref{tab:wutm24}, which simulates WUTM(2,4) transition function;
\item arc multiplicity of subnets $MA5LR$, $MD5LR$ is equal to 5 (instead of 4 in $MA4LR$, $MD4LR$);
\item to simulate peculiarities of weakly universal TM work, two transitions $lb$ and $rb$ are appended which add the blank word codes ${{sw}_{{l}}}$ and ${{sw}_{{r}}}$ to the codes of the left and right parts of tape $L$ and $R$ correspondingly when its value is equal to zero.
\end{itemize}

At the beginning of each computation step, place $STEP$ launches subnet $FS$, which simulates \break WUTM(2,4) transition function $f$ given in  Table~\ref{tab:wutm24}. Subnet $FS$ produces the encoding $s(u')$ of the new state and the encoding $s(x')$ of the new symbol in places $U$ and $X$ respectively to simulate the Turing machine instruction. Subnet $FS$ also puts a token into place $RIGHT$ if the simulated machine instruction is a right move instruction. Place $MOVE$ launches the sequence of subnets $MA5LR$, $MD5LR$, which simulates the control head moves, after subnet $FS$ has finished. At the end of a simulated computation step a token is put into place $STEP$ that allows the simulation of the next instruction to begin.

\begin{figure}
  \centering
    \includegraphics [width=0.6\textwidth] {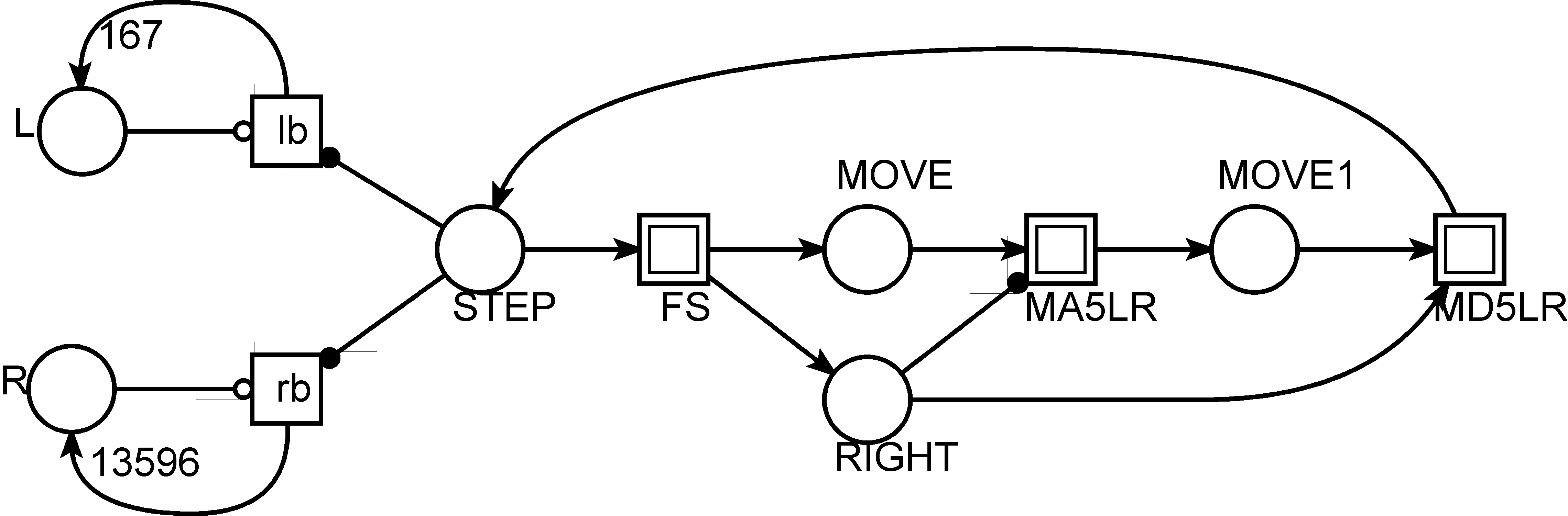}
  \caption{General arrangement of UPN(14,29)}
  \label{fig:upn1429}
\end{figure}

Subnet $FS$ (shown in  Fig.~\ref{fig:fs}) simulates the transition function of WUTM(2,4) as follows: the Turing machine instruction for each pair $(x,u)$ from Table~\ref{tab:wutm24} is encoded by a transition with the label $(s(x), s(u))$ in the $FS$ net; its input arcs from places $X$ and $U$ have corresponding multiplicity, zero multiplicity means the arc absence. For a TM instruction ${(x,u,x',v,u')}$, output arcs to places $X$ and $U$ have multiplicity ${s(x')}$ and ${s(u')}$ correspondingly with an extra arc to place $RIGHT$ when ${v=right}$. Subnet $FS$ of UPN(14,29) is represented in Fig.~\ref{fig:fs}; besides, each transition has an incoming arc from place $STEP$ and outgoing arc to place $MOVE$ which are not shown explicitly. Transitions are not enumerated yet, so the required relations of priorities are shown via arcs connecting transitions, providing only one firable transition on each step of TM simulation.

\begin{figure}
  \centering
    \includegraphics [width=0.7\textwidth] {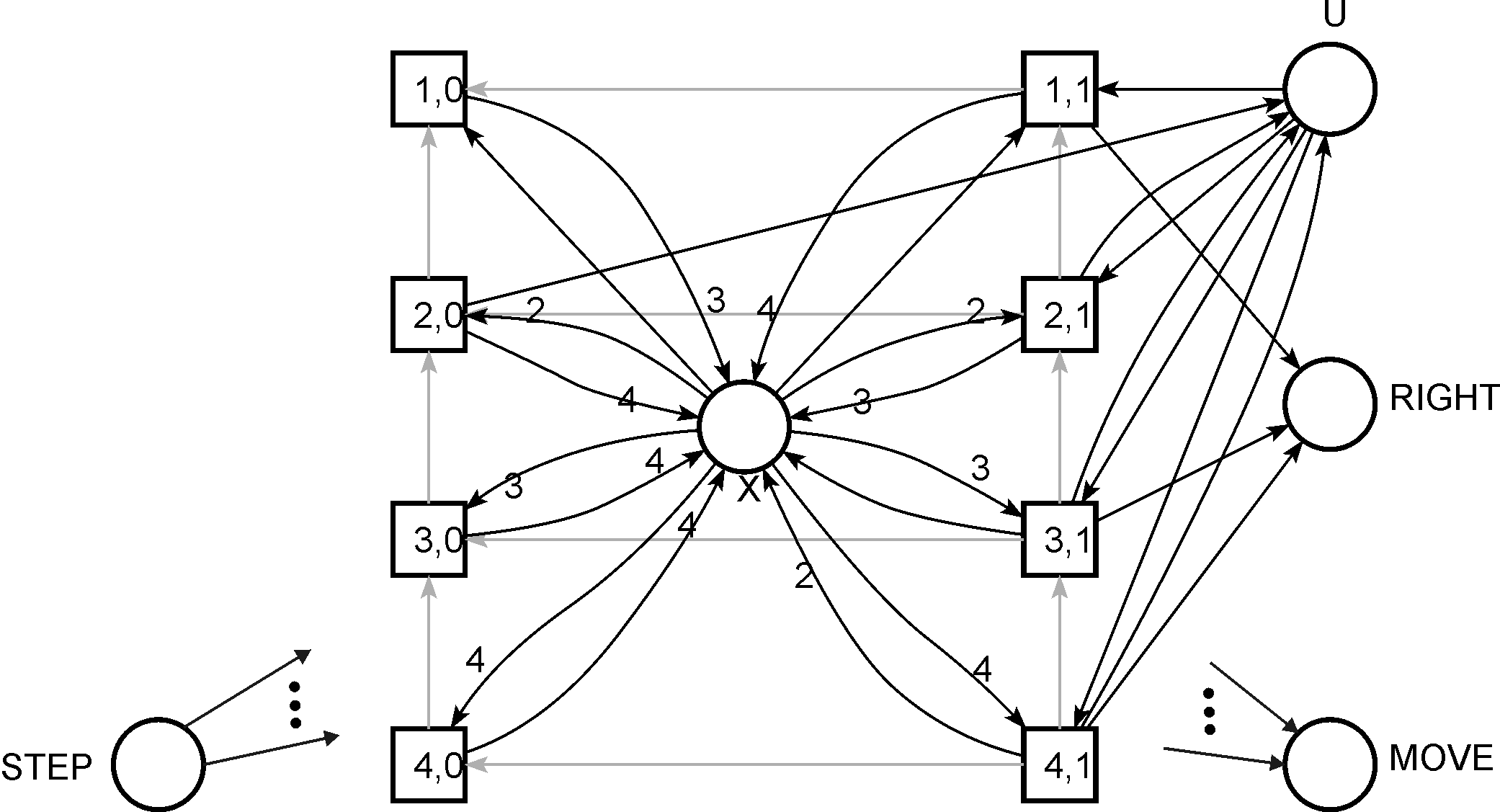}
  \caption{Subnet $FS$ simulating WUTM(2,4) transition function}
  \label{fig:fs}
\end{figure}

UPN(14,29) is composed according to Fig.~\ref{fig:upn1429} via inserting subnets and merging places with the same names. Transitions are enumerated to provide required relations of priorities; places are enumerated in an arbitrary order; the number of arcs is 138. The obtained UPN(14,29) is shown in Fig.~\ref{fig:upn1429full} and Table~\ref{tab:upn1429}; either of them specifies the net but graphical representation is rather tangled. Note that, arcs, connecting transitions, are redundant and could be omitted. Table~\ref{tab:upn1429} entry of form ${x,y}$ specifies the multiplicity of transition's incoming and outgoing arcs correspondingly for each transition-place pair. Zero multiplicity means absence of arcs; inhibitor arcs are represented by the value ${-1}$; empty entries mean two zeroes (no arcs). 

\begin{figure}
  \centering
    \includegraphics [width=0.8\textwidth] {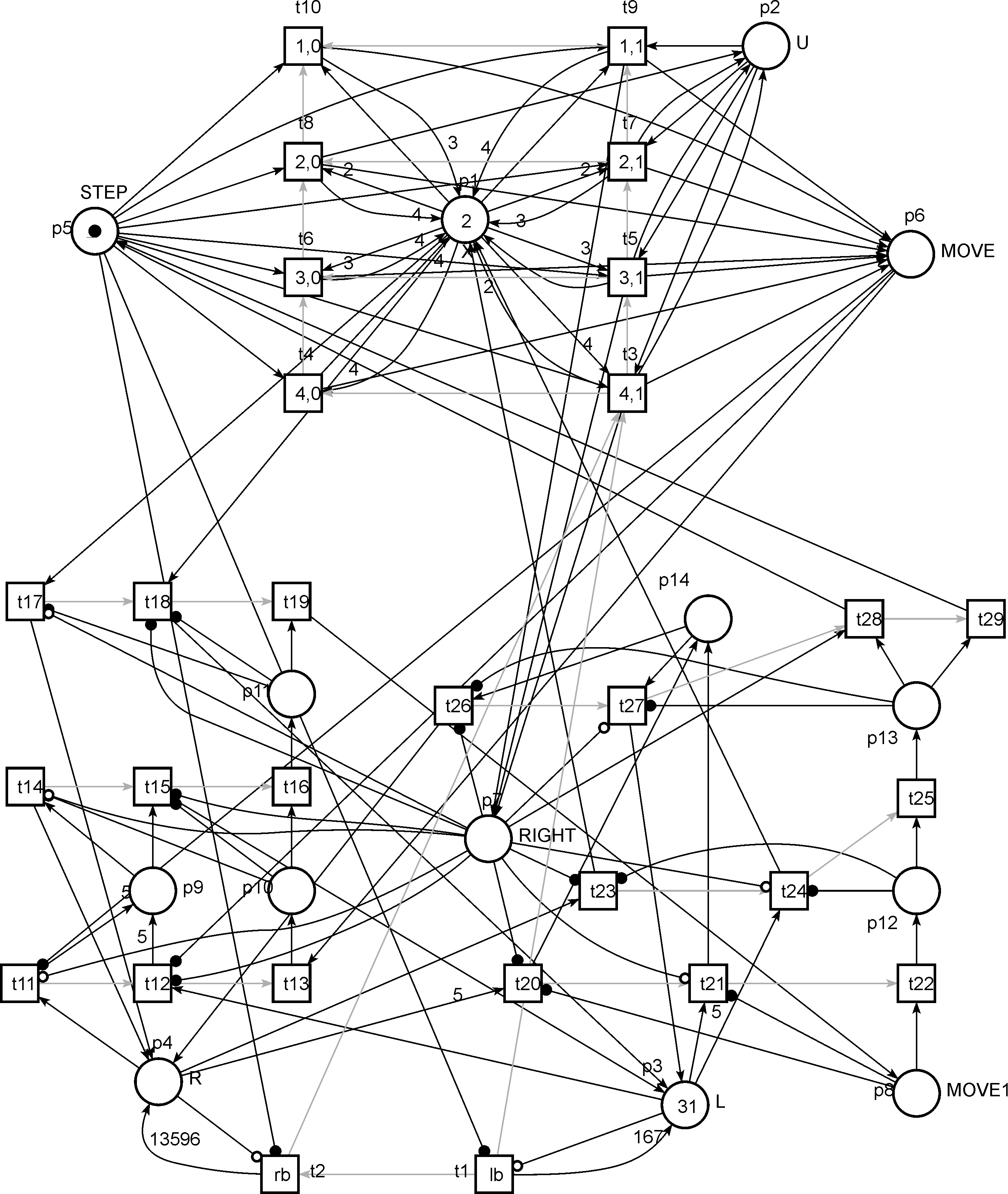}
  \caption{UPN(14,29) in graphical form}
  \label{fig:upn1429full}
\end{figure}

The movement of the tape head on the tape is simulated by the two connected subnets $MA5LR$ and $MD5LR$ inserted in the bottom part of Fig.~\ref{fig:upn1429full}. The meaning of subnets' names is the following: $MA5$ multiplication and addition with radix 5 (${S:=S\cdot 5+X}$), $MD5$ modulo and division with radix 5 (${S:=S~div~5}$, ${X:=S~mod~5 }$); $LR$ choice of places either $L$ or $R$, where codes of the left and right parts of the tape, regarding the current cell symbol code $X$, are stored, depending on the marking of place $RIGHT$. In other words, subnet $MA5$ is used to simulate adding a symbol $x$ to either the left or right tape sequence, while subnet $MD5$ simulates removing a symbol $x$ from either the left or right tape sequence. $MA5$ and $MD5$ are represented in Fig.~\ref{fig:ma5md5}. 

Thus, the sequence of subnets $MA5LR$, $MD5LR$ implements the following operations: 
\begin{itemize}
\item to simulate a left move(when place ${RIGHT=0}$): ${R:=R\cdot 5+X}$, ${L:=L~{div}~5}$, ${X:=L~{mod}~5}$;
\item to simulate a right move(when place ${RIGHT=1}$): ${L:=L\cdot 5+X}$, ${R:=R~{div}~5}$, ${X:=R~{mod}~5}$.
\end{itemize}

\begin{figure}
  \centering
    \includegraphics [width=0.6\textwidth] {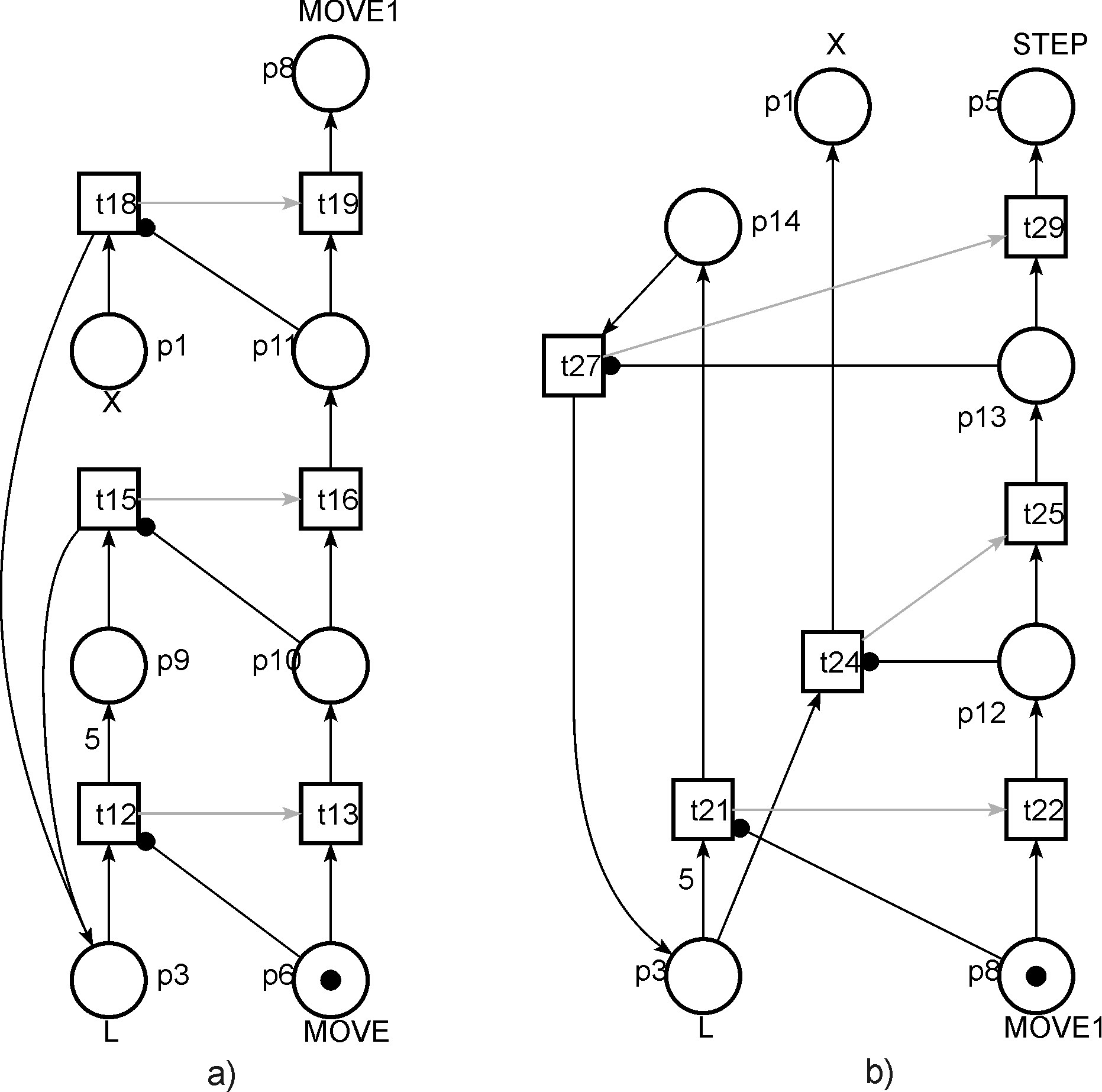}
  \caption{Basic subnets of the tape encoding end decoding: a) add a symbol to the code $MA5$; b) extract a symbol from the code $MD5$}
  \label{fig:ma5md5}
\end{figure}

In Fig.~\ref{fig:ma5md5}, the place and transition numbers were preserved in accordance with Fig.~\ref{fig:upn1429full}. For instance, subnet $MA5$ calculates ${{34}=6\cdot 5+4}$ via the following sequence of transitions' firing \break ${(t_{12})^{6}t_{13}(t_{15})^{30}t_{16}(t_{18})^{4}t_{19}}$ moving the subnet from its initial marking ${L=6,X=4,MOVE=1}$ to the final marking ${L=34,X=0,MOVE1=1}$. Subnet $MD5$ calculates ${6=34~{div}~5}$, ${4=34~{mod}~5}$ via the following sequence of transitions' firing ${(t_{21})^{6}t_{22}(t_{24})^{4}t_{25}(t_{27})^{6}t_{29}}$ moving the subnet from its initial marking ${L=34,X=0,MOVE1=1}$ to the final marking ${L=6,X=4,STEP=1}$.

When constructing $MA5LR$ on $MA5$, transitions ${t_{12},t_{15},t_{18}}$ work with the code $L$ while their twins ${t_{11},t_{14},t_{17}}$ correspondingly work with the code $R$; the choice in controlled by place $RIGHT$. Similarly, when constructing $MD5LR$ on $MD5$, transitions ${t_{21},t_{24},t_{27}}$ work with the code $L$ while their twins ${t_{20},t_{23},t_{26}}$ correspondingly work with the code $R$; the choice is controlled by place $RIGHT$; moreover, transition ${t_{28}}$ (which is a twin of ${t_{29}}$) finally cleans place ${RIGHT}$ in case it contained a token. 
\begin{lemma}
Sequence of subnets $MA5LR$, $MD5LR$, supplied with transitions $lb$, $rb$ having higher priority (Fig.~\ref{fig:upn1429}) and encoding tape symbols starting from unit, simulates work with weakly universal TM tape.
\label{lem:tmtape}
\end{lemma}
\begin{proof}
After the sequence $MA5LR$, $MD5LR$, only one of places $L$, $R$ can become zero: $L$ -- for the left move and $R$ -- for the right move as a result of division operation that means hitting the corresponding edge of the tape working zone. Then place $STEP$ enables one of transition $lb$, $rb$ which fires before $FS$ work and inserts the corresponding blank word code disabling the fired transition. Then, as $lb$, $rb$ are disabled, subnet $FS$ starts. 
\end{proof}

As far as the check on zero is implemented before the current step, blank words are not encoded in the initial configuration.

\begin{theorem}
UPN(14,29) simulates WUTM(2,4) in time ${O(k\cdot5^{{k}})}$ and space ${O(k)}$, where ${k}$ is the number of WUTM(2,4) steps.
\label{the:upnwutm}
\end{theorem}

Theorem~\ref{the:upnwutm} is an immediate conclusion of lemma~\ref{lem:tmtape} and analogous theorem proven in \cite{ZV13} for \break UPN(14,42). The radix for tape encoding is 5, so, in the worst case, on each step, TM goes to the left (right) and after insertion of a blank word, the working zone width is ${k+4}$ (${k+6}$). In the worst case, insertion of a blank word to the left (right) requires an extra transition firing on each of ${k/4}$ (${k/6}$) steps. Considered peculiarities of UPN(14,29) bring constant summands and multipliers (which were omitted) to the estimations.

\begin{table}
  \begin{center} 
\begin{scriptsize}
    \begin{tabular}{|c|c|c|c|c|c|c|c|c|c|c|c|c|c|c|c|}
      \hline
      Sub & $T \backslash P$ & $p_1$ & $p_2$ & $p_3$ & $p_4$ & $p_5$ & $p_6$ & $p_7$ & $p_8$ & $p_9$ & $p_{10}$ & $p_{11}$ & $p_{12}$ & $p_{13}$ & $p_{14}$ \\
      net &  & $X$ & $U$ & $L$ & $R$ & $STEP$ & $MOVE$ & $RIGHT$ & $MOVE1$ &  &  &  &  &  &  \\
     \hline
      $lb$ & $t_1$ &  &  & $-1,$  &  & $1,1$ &  &  &  &  &  &  &  &  &  \\ 
           &       &  &  & $167$ &  &       &  &  &  &  &  &  &  &  &  \\ 
     \hline
      $rb$ & $t_2$ &  &  &  & $-1,$   & $1,1$ &  &  &  &  &  &  &  &  &  \\
           &       &  &  &  & $13596$ &       &  &  &  &  &  &  &  &  &  \\  
     \hline
           & $t_3$ & $4,2$ & $1,1$ &  &  & $1,0$ & $0,1$ & $0,1$ &  &  &  &  &  &  &  \\
     \cline{2-16}
           & $t_4$ & $4,4$ &  &  &  & $1,0$ & $0,1$ &  &  &  &  &  &  &  &  \\
     \cline{2-16}
           & $t_5$ & $3,1$ & $1,1$ &  &  & $1,0$ & $0,1$ & $0,1$ &  &  &  &  &  &  &  \\
     \cline{2-16}
       $F$ & $t_6$ & $3,4$ &  &  &  & $1,0$ & $0,1$ &  &  &  &  &  &  &  &  \\
     \cline{2-16}
       $S$ & $t_7$ & $2,3$ & $1,1$ &  &  & $1,0$ & $0,1$ &  &  &  &  &  &  &  &  \\
     \cline{2-16}
           & $t_8$ & $2,4$ & $0,1$ &  &  & $1,0$ & $0,1$ &  &  &  &  &  &  &  &  \\
     \cline{2-16}
           & $t_9$ & $1,4$ & $1,0$ &  &  & $1,0$ & $0,1$ & $0,1$ &  &  &  &  &  &  &  \\
     \cline{2-16}
           & $t_{10}$ & $1,3$ &  &  &  & $1,0$ & $0,1$ &  &  &  &  &  &  &  &  \\
     \hline
           & $t_{11}$ &  &  &  & $1,0$ &  & $1,1$ & $-1,0$ &  & $0,5$ &  &  &  &  &  \\
     \cline{2-16}
           & $t_{12}$ &  &  & $1,0$ &  &  & $1,1$ & $1,1$ &  & $0,5$ &  &  &  &  &  \\
     \cline{2-16}
       $M$ & $t_{13}$ &  &  &  &  &  & $1,0$ &  &  &  & $0,1$ &  &  &  &  \\
     \cline{2-16}
       $A$ & $t_{14}$ &  &  &  & $0,1$ &  &  & $-1,0$ &  & $1,0$ & $1,1$ &  &  &  &  \\
     \cline{2-16}
       $5$ & $t_{15}$ &  &  & $0,1$ &  &  &  & $1,1$ &  & $1,0$ & $1,1$ &  &  &  &  \\
     \cline{2-16}
       $L$ & $t_{16}$ &  &  &  &  &  &  &  &  &  & $1,0$ & $0,1$ &  &  &  \\
     \cline{2-16}
       $R$ & $t_{17}$ & $1,0$ &  &  & $0,1$ &  &  & $-1,0$ &  &  &  & $1,1$ &  &  &  \\
     \cline{2-16}
           & $t_{18}$ & $1,0$ &  & $0,1$ &  &  &  & $1,1$ &  &  &  & $1,1$ &  &  &  \\
     \cline{2-16} 
           & $t_{19}$ &  &  &  &  &  &  &  & $0,1$ &  &  & $1,0$ &  &  &  \\
     \hline
           & $t_{20}$ &  &  &  & $5,0$ &  &  & $1,1$ & $1,1$ &  &  &  &  &  & $0,1$ \\
     \cline{2-16}
           & $t_{21}$ &  &  & $5,0$ &  &  &  & $-1,0$ & $1,1$ &  &  &  &  &  & $0,1$ \\
     \cline{2-16}
       $M$ & $t_{22}$ &  &  &  &  &  &  &  & $1,0$ &  &  &  & $0,1$ &  &  \\
     \cline{2-16}
       $D$ & $t_{23}$ & $0,1$ &  &  & $1,0$ &  &  & $1,1$ &  &  &  &  & $1,1$ &  &  \\
     \cline{2-16}
       $5$ & $t_{24}$ & $0,1$ &  & $1,0$ &  &  &  & $-1,0$ &  &  &  &  & $1,1$ &  &  \\
     \cline{2-16}
       $L$ & $t_{25}$ &  &  &  &  &  &  &  &  &  &  &  & $1,0$ & $0,1$ &  \\
     \cline{2-16}
       $R$ & $t_{26}$ &  &  &  & $0,1$ &  &  & $1,1$ &  &  &  &  &  & $1,1$ & $1,0$ \\
     \cline{2-16}
           & $t_{27}$ &  &  & $0,1$ &  &  &  & $-1,0$ &  &  &  &  &  & $1,1$ & $1,0$ \\
     \cline{2-16}
           & $t_{28}$ &  &  &  &  & $0,1$ &  & $1,0$ &  &  &  &  &  & $1,0$ &  \\
     \cline{2-16}
           & $t_{29}$ &  &  &  &  & $0,1$ &  &  &  &  &  &  &  & $1,0$ &  \\
     \hline
    \end{tabular}
\end{scriptsize}
  \end{center}
  \caption{UPN(14,29) in tabular form}
  \label{tab:upn1429}
\end{table}

An example of UPN(14,29) work, on the sequence of WUTM(2,4) configurations described in \cite{NE08,NW09W}, is represented in Table~\ref{tab:upn1429tr}. The codes of state and tape are shown for markings when place STEP contains a token (before running current TM step simulation). Initial tape code corresponds to the word ${{000}}$ on the left and the current cell symbol ${1}$. Before executing the step simulation, transition $rb$ fires and puts the code of the right blank word ${{sw}_{r}={13596}}$ into place ${R}$. A similar situation arises after step 5 and the code of the left blank word ${{sw}_{l}={167}}$ is put by transition $lb$ into place $L$. The part of the encoded TM configuration (the tape working zone) is highlighted in a bold font. The decoded result on each step completely corresponds with configurations considered in \cite{NE08,NW09W}. 

\begin{table}
  \begin{center}
    \begin{tabular}{|c|c|c|c|}
      \hline  
           &               & \multicolumn{2}{c|}{Code of} \\
      \cline{3-4}
      Step & Configuration & state & tape \\
           &  & $U$ & $(L,X,R)$ \\
      \hline
       0 & $u_1,...00{\tnMarkZero}1~00{\tnMarkZero}1~\textbf{000\underline{1}}~0{\tnMarkOne}{\tnMarkZero}{\tnMarkZero}0{\tnMarkOne}~ 0{\tnMarkOne}{\tnMarkZero}{\tnMarkZero}0{\tnMarkOne} ...$ & 0 & (31,2,0) \\
      \hline
         & $u_1,...00{\tnMarkZero}1~00{\tnMarkZero}1~\textbf{000\underline{1}~0{\tnMarkOneB}{\tnMarkZeroB}{\tnMarkZeroB}0{\tnMarkOneB}}~0{\tnMarkOne}{\tnMarkZero}{\tnMarkZero}0{\tnMarkOne} ...$ & 0 & (31,2,13596) \\
      \hline
       1 & $u_2,...00{\tnMarkZero}1~00{\tnMarkZero}1~\textbf{00\underline{0}{\tnMarkOneB}~0{\tnMarkOneB}{\tnMarkZeroB}{\tnMarkZeroB}0{\tnMarkOneB}}~0{\tnMarkOne}{\tnMarkZero}{\tnMarkZero}0{\tnMarkOne} ...$ & 1 & (6,1,67984) \\
      \hline
       2 & $u_1,...00{\tnMarkZero}1~00{\tnMarkZero}1~\textbf{00{\tnMarkOneB}\underline{\tnMarkOneB}~0{\tnMarkOneB}{\tnMarkZeroB}{\tnMarkZeroB}0{\tnMarkOneB}}~0{\tnMarkOne}{\tnMarkZero}{\tnMarkZero}0{\tnMarkOne} ...$ & 0 & (34,4,13596)\\
      \hline
       3 & $u_1,...00{\tnMarkZero}1~00{\tnMarkZero}1~\textbf{00\underline{\tnMarkOneB}{\tnMarkOneB}~0{\tnMarkOneB}{\tnMarkZeroB}{\tnMarkZeroB}0{\tnMarkOneB}}~0{\tnMarkOne}{\tnMarkZero}{\tnMarkZero}0{\tnMarkOne} ...$ & 0 & (6,4,67984) \\
      \hline
       4 & $u_1,...00{\tnMarkZero}1~00{\tnMarkZero}1~\textbf{0\underline{0}{\tnMarkOneB}{\tnMarkOneB}~0{\tnMarkOneB}{\tnMarkZeroB}{\tnMarkZeroB}0{\tnMarkOneB}}~0{\tnMarkOne}{\tnMarkZero}{\tnMarkZero}0{\tnMarkOne} ...$ & 0 & (1,1,339924) \\
      \hline 
       5 & $u_1,...00{\tnMarkZero}1~00{\tnMarkZero}1~\textbf{\underline{0}{\tnMarkZeroB}{\tnMarkOneB}{\tnMarkOneB}~0{\tnMarkOneB}{\tnMarkZeroB}{\tnMarkZeroB}0{\tnMarkOneB}}~0{\tnMarkOne}{\tnMarkZero}{\tnMarkZero}0{\tnMarkOne} ...$ & 0 & (0,1,1699623) \\
      \hline
         & $u_1,...00{\tnMarkZero}1~\textbf{00{\tnMarkZeroB}1~\underline{0}{\tnMarkZeroB}{\tnMarkOneB}{\tnMarkOneB}~0{\tnMarkOneB}{\tnMarkZeroB}{\tnMarkZeroB}0{\tnMarkOneB}}~0{\tnMarkOne}{\tnMarkZero}{\tnMarkZero}0{\tnMarkOne} ...$ & 0 & (167,1,1699623) \\
      \hline
       6 & $u_1,...00{\tnMarkZero}1~\textbf{00{\tnMarkZeroB}\underline{1}~{\tnMarkZeroB}{\tnMarkZeroB}{\tnMarkOneB}{\tnMarkOneB}~0{\tnMarkOneB}{\tnMarkZeroB}{\tnMarkZeroB}0{\tnMarkOneB}}~0{\tnMarkOne}{\tnMarkZero}{\tnMarkZero}0{\tnMarkOne} ...$ & 0 & (33,2,8498118) \\
      \hline
       7 & $u_2,...00{\tnMarkZero}1~\textbf{00\underline{\tnMarkZeroB}{\tnMarkOneB}~{\tnMarkZeroB}{\tnMarkZeroB}{\tnMarkOneB}{\tnMarkOneB}~0{\tnMarkOneB}{\tnMarkZeroB}{\tnMarkZeroB}0{\tnMarkOneB}}~0{\tnMarkOne}{\tnMarkZero}{\tnMarkZero}0{\tnMarkOne} ...$ & 1 & (6,3,42490594) \\
      \hline
       8 & $u_2,...00{\tnMarkZero}1~\textbf{000\underline{\tnMarkOneB}~{\tnMarkZeroB}{\tnMarkZeroB}{\tnMarkOneB}{\tnMarkOneB}~0{\tnMarkOneB}{\tnMarkZeroB}{\tnMarkZeroB}0{\tnMarkOneB}}~0{\tnMarkOne}{\tnMarkZero}{\tnMarkZero}0{\tnMarkOne} ...$ & 1 & (31,4,8498118) \\
      \hline
       9 & $u_2,...00{\tnMarkZero}1~\textbf{0001~\underline{\tnMarkZeroB}{\tnMarkZeroB}{\tnMarkOneB}{\tnMarkOneB}~0{\tnMarkOneB}{\tnMarkZeroB}{\tnMarkZeroB}0{\tnMarkOneB}}~0{\tnMarkOne}{\tnMarkZero}{\tnMarkZero}0{\tnMarkOne} ...$ & 1 & (157,3,1699623) \\
      \hline
       10 & $u_2,...00{\tnMarkZero}1~\textbf{0001~0\underline{\tnMarkZeroB}{\tnMarkOneB}{\tnMarkOneB}~0{\tnMarkOneB}{\tnMarkZeroB}{\tnMarkZeroB}0{\tnMarkOneB}}~0{\tnMarkOne}{\tnMarkZero}{\tnMarkZero}0{\tnMarkOne} ...$ & 1 & (786,3,339924) \\
      \hline
       11 & $u_2,...00{\tnMarkZero}1~\textbf{0001~00\underline{\tnMarkOneB}{\tnMarkOneB}~0{\tnMarkOneB}{\tnMarkZeroB}{\tnMarkZeroB}0{\tnMarkOneB}}~0{\tnMarkOne}{\tnMarkZero}{\tnMarkZero}0{\tnMarkOne} ...$ & 1 & (3931,4,67984) \\
      \hline
       12 & $u_2,...00{\tnMarkZero}1~\textbf{0001~001\underline{\tnMarkOneB}~0{\tnMarkOneB}{\tnMarkZeroB}{\tnMarkZeroB}0{\tnMarkOneB}}~0{\tnMarkOne}{\tnMarkZero}{\tnMarkZero}0{\tnMarkOne} ...$ & 1 & (19657,4,13596) \\
      \hline
       13 & $u_2,...00{\tnMarkZero}1~\textbf{0001~0011~\underline{0}{\tnMarkOneB}{\tnMarkZeroB}{\tnMarkZeroB}0{\tnMarkOneB}}~0{\tnMarkOne}{\tnMarkZero}{\tnMarkZero}0{\tnMarkOne} ...$ & 1 & (98287,1,2719) \\
      \hline
       14 & $u_1,...00{\tnMarkZero}1~\textbf{0001~0011~{\tnMarkOneB}\underline{\tnMarkOneB}{\tnMarkZeroB}{\tnMarkZeroB}0{\tnMarkOneB}}~0{\tnMarkOne}{\tnMarkZero}{\tnMarkZero}0{\tnMarkOne} ...$ & 0 & (491439,4,543) \\
      \hline
       ... & ... & ... & ... \\
      \hline
    \end{tabular}
  \end{center}
  \caption{Trace of UPN(14,29) running}
  \label{tab:upn1429tr}
\end{table}

\section{Simulating BTS by TM}

In \cite{NW09FI} BTS was constructed, as an intermediate system, on TM via cyclic TM. Here we solve an inverse task that is rather simple but required to supply a missed link in the chosen chain of DIPN encoding described in section 3.

Let BTS ${B=(A,E,e_{{h}},R)}$ be given. We construct TM ${{MB}=(\Omega,\Sigma,f,{qs},{qh})}$, where states are denoted with letters $"q"$ to distinguish $MB$ from WUTM(2,4) and

{\centering
${\Omega={qs}\cup {qh}\cup {ql}\cup \{q_{{a}}|a\in A\}\cup \{q_{{e}}|e\in E\}\cup}$

${\cup \{q_{{e,a}}|e\in E,a\in A\}\cup \{\mathit{q1}_{{e,a}}|e\in E,a\in A\}\cup \{\mathit{q2}_{{e,a}}|e\in E,a\in
A\}}$,

${\Sigma =\lambda \cup A\cup E}$.
\par}

The BTS initial configuration is written directly on the $MB$ tape with the control head in state $qs$ on its leftmost symbol. Apart from the start $qs$ and halt $qh$ states, $MB$ contains:
\begin{itemize}
\item a state $ql$ for the left move to the BTS configuration head after insertion into its tail;
\item a state ${q_{{a}}}$ for each symbol ${a\in A}$ to store the first symbol ${a\in A}$ of the BTS configuration after its deletion to bring it to the tail according to productions of form ${R(a)=a}$;
\item a state ${q_{{e}}}$ for each symbol ${e\in E}$ to store the first symbol ${e\in E}$ of the BTS configuration after its deletion to process the following symbol ${a\in A}$ according to productions of forms ${R(e,a)=AE}$ and ${R(e,a)=AAE}$;
\item a state ${q_{{e,a}}}$ for each pair of symbols ${e\in E}$, ${a\in A}$ to store the current production after deletion of symbols' pair ${ea}$ from the beginning of the BTS configuration;
\item states ${\mathit{q1}_{{e,a}}}$, ${\mathit{q2}_{{e,a}}}$ for each pair of symbols ${e\in E}$, ${a\in A}$ for insertion of the second and third symbols of the right part of productions of forms ${R(e,a)=AE}$ and ${R(e,a)=AAE}$.
\end{itemize}
The $MB$ transition function ${f}$ is described in Table~\ref{tab:mb}.

\begin{table}
  \begin{center}
    \begin{tabular}{|l|l|l|}
      \hline
      No & Instruction & Description \\
      \hline
       $1.1$ & $(qs,e_h,qh,S,e_h);$ & halt \\
      \hline
       $2.1$ & $(qs,a,q_a,R,\lambda),$ &  \\
       $2.2$ & $(q_a,x,q_a,R,x),~x\neq \lambda,$ & bring $a$ to the tail for $R(a)=a$ \\
       $2.3$ & $(q_a,\lambda,ql,L,a);$ &  \\
      \hline
       $3.1$ & $(qs,e,q_e,R,\lambda),$ &  \\
       $3.2$ & $(q_e,a,q_{e,a},R,\lambda),$ & bring $R(e,a)$ to the tail \\
       $3.3$ & $(q_{e,a},x,q_{e,a},R,x),~x\neq \lambda,$ & for $R(e,a)\in AE$ or $R(e,a)\in AAE$: \\
       $3.4$ & $(q_{e,a},\lambda,q1_{e,a},R,s1_{e,a}),$ & $s1_{e,a}=a'$, $s2_{e,a}=e'$, $s3_{e,a}=\lambda$ for $R(e,a)=a'e'$ \\
       $3.5$ & $(q1_{e,a},\lambda,q2_{e,a},R,s2_{e,a}),$ & $s1_{e,a}=a'$, $s2_{e,a}=a''$, $s3_{e,a}=e'$ for $R(e,a)=a'a''e'$ \\
       $3.6$ & $(q2_{e,a},\lambda,ql,L,s3_{e,a});$ &  \\
      \hline
       $4.1$ & $(ql,x,ql,L,x),~x\neq \lambda,$ & move to the head \\
       $4.2$ & $(ql,\lambda,qs,R,\lambda).$ &  \\
      \hline
    \end{tabular}
  \end{center}
  \caption{Instructions of Turing machine $MB$}
  \label{tab:mb}
\end{table}

\begin{theorem}
Turing machine $MB$ (Table~\ref{tab:mb}) simulates bi-tag system $B$ in time ${O(k^{{2}})}$ with space ${O(k)}$, where ${k}$ is the number of applied productions.
\label{the:mb}
\end{theorem}
\begin{proof}
At first, we show that $MB$ simulates $B$. Its initial tape working zone coincides with the initial configuration of $B$ and the control head is positioned on the leftmost symbol. If ${w=eas'}$ and ${e=e_{{h}}}$, it halts according to instruction 1.1. 

If ${w=as'}$, it stores the first symbol ${a}$ with state ${q_{{a}}}$ and deletes it via instructions 2.1, moves to the tail via instructions 2.2, replaces the first blank symbol ${\lambda }$ with the stored symbol ${a}$ via instructions 2.3 and switches to state ${ql}$ that completely corresponds to the application of production of type ${R(a)=a}$. 

If ${w=eas'}$, it stores the first symbol ${e}$ with state ${q_{{e}}}$, deletes it and moves to the left via instruction 3.1, stores the combination ${(e,a)}$ with state ${q_{{e,a}}}$ and deletes ${a}$ via instructions 3.2, moves to the tail via instructions 3.3, replaces the first blank symbol ${\lambda}$ with symbol ${\mathit{s1}_{{e,a}}}$ via instructions 3.4, replaces the second blank symbol ${\lambda}$ with symbol ${\mathit{s2}_{{e,a}}}$ via instructions 3.5, replaces the third blank symbol ${\lambda}$ with symbol ${\mathit{s3}_{{e,a}}}$ via instructions 3.6 and switches to state $ql$ that completely corresponds to the application of production of either type ${R(e,a)=AE}$ or ${R(e,a)=AAE}$.

After applying a production, the control head is positioned at the rightmost or previous to the rightmost symbol of the working zone in state $ql$. Then instructions 4.1 and 4.2 return it to the leftmost symbol. As there are no other valid sequences except described, $MB$ simulates $B$.

Neglecting the initial BTS configuration length, note that each production application extends it by unit in the worst case, thus space complexity is ${O(k)}$. Implementation of each production requires scanning the current word twice; thus, time complexity is ${O(2\cdot k\cdot k)\approx O(k^{{2}})}$. 
\end{proof}
Combining theorems~\ref{the:upnwutm} and \ref{the:mb} with results on complexity estimations of simulations ${{DIPN}\rightarrow {BTS}}$ \cite{NW09FI} and ${{TM}\rightarrow {WUTM}(2,4)}$ \cite{NE08,NW09W}, we come to the following corollary.
\begin{corollary}
UPN(14,29) simulates a given DIPN ${N}$ in exponential time and polynomial space with respect to the number of ${N}$ steps.
\end{corollary}

Definite forms of complexity functions could be obtained via substitution. $BTS$ simulates $DIPN$ \cite{ZV13} in time ${x=O(k^{{3}})}$ with space ${y=O(k)}$, where ${k}$ is the number of ${N}$ steps. $TM$ simulates $BTS$ (theorem~\ref{the:mb}) in time ${u=O(x^{{2}})}$ with space ${v=O(x)}$. WUPN(2,4) simulates $TM$ \cite{NE08,NW09W} in time ${r=O(u^{{4}}\cdot {log}^{{2}}u)}$. And UPN(14,29) simulates WUPN(2,4) (theorem~\ref{the:upnwutm}) in time ${O(r\cdot 5^{{r}})}$ with space ${O(r)}$. Thus UPN(14,29) simulates a given $DIPN$ in time ${O(z\cdot 5^{{z}})}$ with space ${O(z)}$, where ${z=k^{24}\cdot {log}^{12}k}$.

\section{Conclusions}

Thus, the universal deterministic inhibitor Petri net with 14 places and 29 transitions was constructed based on direct encoding of Neary and Woods weakly universal Turing machine with 2 states and 4 symbols using early presented technique \cite{ZV13}; universal net contains 13 transitions lesser comparing \cite{ZV13}. 

A technique for the blank words simulation of the weakly universal TM tape was developed that consists in adding blank words' codes when reaching zero value of the tape codes, so constructed DIPN is not weak. Moreover, a simulation of a bi-tag system by a TM was developed to fill the gap in the chain of source DIPN translation.

Resulting universal Petri net runs in exponential time and polynomial space with respect to the target DIPN transitions' firing sequence length. Constructing efficient UPN is a direction for future work. But the main obstacle consists in the arithmetic encoding/decoding technique via increment/decrement operations implemented by PN transitions.

\section*{Acknowledgement}

The author would like to thank Turlough Neary for his help in improving the readability of the paper.


\nocite{*}
\bibliographystyle{eptcs}

\end{document}